\documentclass[a4paper,11pt]{amsart}
\usepackage{amsfonts,amssymb,epsfig,latexsym}
\usepackage{amsmath,comment,bm,mathrsfs,braket}
\usepackage[latin1]{inputenc}
\usepackage{color,layout}
\usepackage{ifthen}

\usepackage{xparse}
\makeatletter
\RenewDocumentCommand{\title}{om}{%
   \IfNoValueTF{#1}
     {\gdef\shorttitle{Feynman integral in quantum walk}}%
     {\gdef\shorttitle{#1}}%
   \gdef\@title{#2}%
}
\makeatother
\usepackage{graphicx}
\usepackage{color}  
\usepackage{bmpsize}

\newtheorem{theorem}{Theorem}[section]
\newtheorem{lemma}[theorem]{Lemma}

\newtheorem{definition}[theorem]{Definition}
\newtheorem{remark}[theorem]{Remark}

\def\square{\hbox{\vrule\vbox{\hrule\phantom{o}\hrule}\vrule}}

\newcommand{\be}{\begin{equation}}
\newcommand{\ee}{\end{equation}}

\newcommand{\til}[1]{\widetilde{#1}}

\numberwithin{equation}{section}

\newcommand{\N}{\mathbb{N}}
\newcommand{\Z}{\mathbb{Z}}
\newcommand{\R}{\mathbb{R}}
\newcommand{\C}{\mathbb{C}}

\newcommand{\SM}{\mathbb{S}}
\newcommand{\W}{{\mathcal W}}
\newcommand{\cS}{{\mathcal S}}

\newcommand{\cA}{{\mathcal A}}

\newcommand{\cB}{{\mathcal B}}

\newcommand{\cU}{{\mathcal U}}

\newcommand{\e}{\varepsilon}

\newcommand{\pphi}{\varphi}

\newcommand{\cT}{\mathcal{T}}
\newcommand{\re}{{\rm Re}\hskip 1pt }
\newcommand{\im}{{\rm Im}\hskip 1pt }
\newcommand{\ord}{{\mathcal O}}

\newcommand{\ope}[1]{{\operatorname{#1}}}
\newcommand{\mc}[1]{{\mathcal{#1}}}

\numberwithin{equation}{section}

\begin{document}

\title{Feynman Integral in Quantum Walk,
Barrier-top Scattering and Hadamard Walk 
}
\author{Kenta Higuchi}


\begin{abstract}
The aim of this article is to relate the discrete quantum walk on $\mathbb{Z}$ with the continuous Schr\"odinger operator on $\mathbb{R}$ in the scattering problem. Each point of $\mathbb{Z}$ is associated with a barrier of the potential, and the coin operator of the quantum walk is determined by the transfer matrix between bases of WKB solutions on the classically allowed regions of both sides of the barrier. This correspondence enables us to represent each entry of the scattering matrix of the Schr\"odinger operator as a countable sum of probability amplitudes associated with the paths of the quantum walker. In particular, the barrier-top scattering corresponds to the Hadamard walk in the semiclassical limit.
\end{abstract}
\maketitle

\section{Introduction}
We consider the  one dimensional semiclassical Schr\"odinger equation 
\begin{align}\label{Scheq}
(P-\lambda)\pphi=0,\quad
P(h)=-h^2\frac{d^2}{dx^2} +V(x), \quad x\in \R,
\end{align}
where $h$ is a positive small parameter, $V(x)$ is a  real-valued  function on $\R$. 
Under the short-range condition, the scattering matrix $\SM_\ope{QM}$ is defined as a $2\times2$ matrix which expresses the outgoing waves in terms of the incoming waves. 
Each entry of this matrix, say the $(1,1)$-entry, corresponds to waves coming from $+\infty$ and going to $-\infty$.
It is related to Feynman path integral of the complex probability amplitudes over all  continuous paths from $+\infty$ to $-\infty$. 

In the semiclassical limit, however, it turned out in the barrier-top scattering that the Feynman integral reduces to a countable sum
over classical trajectories reflecting or transmitting at the barrier-tops (\cite{Fu}).

On the other hand, the scattering matrix is defined for the quantum walk on $\Z$, 
and each entry can be computed as a countable sum over all possible paths of the quantum walker of  complex numbers determined by the coin operators that each path passes through (\cite{KKMS}).

A connection  between quantum walk and the Schr\"odinger equation in scattering problem has recently been studied in \cite{HKSS,MMOS}
for point interactions by delta functions.
In this paper, we extend this connection to general short range potentials. On the Schr\"odinger side, we only assume that the potential is continuous in the classically allowed regions. Such a region is interrupted by a finite number of ``barriers", where the potential may present barrier-top, step-like singularity or even delta functions etc.

 In each classically allowed intervals, we construct a pair of solutions to the Schr\"odinger equation (outgoing and incoming WKB solutions). We define a ``local scattering matrix'' at each barrier which expresses the outgoing solutions from this barrier in terms of
 incoming one to this barrier. This matrix is unitary when the wronskian of the pair does not depend on the interval.
 Then we obtain a quantum walk on $\Z$ with these local scattering matrices as coin operators which define the time evolution $\cU$. 
 
Under this correspondence, a state $\Psi$ of this quantum walk is stationary;
\be\label{eq:StationaryQW}
\cU\Psi=\Psi,
\ee
if and only if each element of $\Psi$ gives the linear combination of solutions of the stationary Schr\"odinger equation in the corresponding interval (Lemma \ref{thm:QWQM1}), and it follows that the scattering matrix of the Schr\"odinger equation and that of the quantum walk coincide  (Theorem \ref{thm:SMatrix}). 
This enables us to compute the scattering matrix of the Schr\"odinger equation applying the counting path method in the quantum walk. 

An interesting example is the barrier-top scattering since the quantum tunneling effect is not negligible in this case.
The semiclassical asymptotic behavior of the transfer matrix, and hence the local scattering matrix at each barrier-top is already known
(\cite{Ra,FR,Fu, CPS}) and we recover the result of \cite{Fu} which gives a Feynman representation of the scattering matrix in the semiclassical limit.
From our point of view, we can say that the barrier-top scattering 
corresponds to the Hadamard walk in the semiclassical limit
 (see Sec.~\ref{Sec:BTandHW}). 



\section{Feynman integral in quantum walk}
Here we give an interpretation of quantum walks as a projection of the scattering problems of the Schr\"odinger equation. 
Then the Feynman integral in quantum walks is given by a countable sum by this interpretation.  
\subsection{Scattering problem for the Schr\"odinger equation}\label{Sec:Sch}
We consider scattering problems for the one dimensional Schr\"odinger equation \eqref{Scheq} with a continuous real valued potential $V$. 
We assume the following so-called short range condition (V1):

\vspace{0.2cm}
\noindent
\textbf{(V1)}
There exist $C>0$, $\epsilon>0$ such that $\left|V(x)\right|\le C(1+\left|x\right|)^{-1-\epsilon}$. 

Let $\lambda$ be positive. Under (V1), there exist Jost solutions $J_\ope{in}^\pm (x),$ $J_\ope{out}^\pm (x)$ characterized by the asymptotic behavior 
as $x\to\pm\infty$ (see e.g., \cite{ReSi}):
\begin{align*}
&e^{i\sqrt{\lambda}\, \left|x\right|/h}
J_\ope{in}^\pm \to1,\quad e^{-i\sqrt{\lambda}\, \left|x\right|/h}
J_\ope{out}^\pm \to1.
\end{align*}
Remark that $J_\ope{in}^\pm=\overline{J_\ope{out}^\pm}$. 
The pairs $(J_\ope{in}^+, J_\ope{out}^+)$ and $(J_\ope{out}^-, J_\ope{in}^-)$ make bases of solutions to \eqref{Scheq}, and the  matrix $\mathbb T$ changing the basis
\be
(J_\ope{in}^+,J_\ope{out}^+)=(J_\ope{out}^-,J_\ope{in}^-)\mathbb T,
\ee
belongs to the subgroup $\cT$ of $\ope{SL}(2,\C)$:
\be\label{eq:formTn}
\cT:=\{T=(t_{j,k})\in \ope{SL}(2,\C);\,t_{11}=\overline{t_{22}},\,t_{12}=\overline{t_{21}} \}.
\ee
The scattering matrix $\SM_{\ope{QM}}$ is the $2\times2$ matrix defined by
\be\label{eq:SMQM}
(J_\ope{in}^+,J_\ope{in}^-)=(J_\ope{out}^-,J_\ope{out}^+)\SM_\ope{QM}.
\ee
It is a unitary matrix
with the same non-vanishing diagonal entries i.e. an element of the set
\be
\cS:=\{M=(m_{j,k})\in { U}(2);\,m_{11}=m_{22}\neq0\},
\ee 
and obtained from $\mathbb T$ by the map $\mc{M}$ from $\cT$ to $\cS$ defined by
\be\label{eq:DefM}
\mc{M}\begin{pmatrix}p&\bar{q}\\q&\bar{p}\end{pmatrix}
:=\frac{1}{\bar{p}}\begin{pmatrix}1&\bar{q}\\-q&1\end{pmatrix},
\ee
This map is bijective and the inverse is given by
$$
\mc{M}^{-1}\begin{pmatrix}a&b\\c&a\end{pmatrix}
=\begin{pmatrix}\bar{a}^{-1}&a^{-1}b\\\bar{a}^{-1}\bar{b}&a^{-1}\end{pmatrix}
=\begin{pmatrix}\bar{a}^{-1}&-\bar{a}^{-1}\bar{c}\\-a^{-1}c&a^{-1}\end{pmatrix}.
$$
The set $\cS$ can be regarded as a group with the operation $*$ induced from $\cT$:
\be
S_1*S_2:=\mc{M}((\mc{M}^{-1}S_1)(\mc{M}^{-1}S_2)),\quad
S_j\in\cS.
\ee
\begin{remark}
The definition \eqref{eq:SMQM} of the scattering matrix is slightly modified from the usual one, $(J_\ope{in}^-,J_\ope{in}^+)=(J_\ope{out}^+,J_\ope{out}^-)\til{\SM}_\ope{QM}$. 
\end{remark}

Now we fix  a closed interval $\Lambda=[\lambda_0,\lambda_1]$ with $0<\lambda_0\le \lambda_1$.
The condition (V1) implies that the classically forbidden region $\{x\in\R;\,\lambda_0\le V(x)\}$ consists of a finite number of closed intervals $(K_n)_{n=1}^{n_0}$, $K_n=[x_{2n-1},x_{2n}]$, with a strictly increasing sequence $(x_n)_{n=1}^{2n_0}\subset\R$.

The complement of $(K_n)_{n=1}^{n_0}$ consists of classically allowed intervals $I_n:=(x_{2n},x_{2n+1})$ for $n=1,2,\ldots,n_0-1$, $I_0:=(-\infty,x_1)$ and $I_{n_0}:=(x_{2n_0},+\infty)$. 
On each $I_n$, there exists a solution $g_n$ to the Schr\"odinger equation \eqref{Scheq} for $\lambda \in\Lambda$ such that the pair $(g_n,\bar{g}_n)$ forms a basis of the solutions. Here $\bar{g}_n$ is the complex conjugate of $g_n$. Denote $g_{0}=J_\ope{out}^-$ and $g_{n_0}=J_\ope{in}^+$,
and let $T_n$ be the $2\times 2$ matrix defined for $n=1,2,\ldots,n_0$ by 
\be\label{eq:DefTn}
(g_{n-1},\bar{g}_{n-1})T_n=(g_n,\bar{g}_n).
\ee
Then we have ${\mathbb T}=T_1T_2\cdots T_{n_0}$.
We assume without loss of generality that the Wronskians are independent of $n\in\Z$ and that the pair $(g_{n-1},\bar{g}_n)$ is linearly independent for any $n\in\Z$. Then each transfer matrix $T_n$ belongs to $\cT$.
The scattering matrix can be obtained by the image of the product of the transfer matrices $T_n$ by $\mc{M}$, 
\be\label{eq:SMandTn}
\SM_\ope{QM}=\mc{M}({\mathbb T})
=\mc{M}(T_1)*\mc{M}(T_2)*\cdots*\mc{M}(T_{n_0}).
\ee 
Note that we have $(g_n,\bar{g}_{n-1})=(g_{n-1},\bar{g}_n)\mc{M}(T_n)$. 

\subsection{Discrete time quantum walk on $\Z$}\label{Sec:DTQWZ}
We recall a standard definition of the two-state quantum walk model on $\Z$ (see e.g. \cite{MMOS}). 
A quantum state is a bounded sequence $\Psi=(\Psi(n))_{n\in\Z}\in\cB:=l^\infty(\Z;\C^2)$ of column vectors $\Psi(n)\in\C^2$, and 
its discrete time evolution is given by a bounded linear operator $\cU$ on the Banach space $\cB$. This operator $\cU$ is determined by a sequence $(U_n)_{n\in\Z}$ of $2\times2$ unitary matrices in the following way: 

\begin{definition}\label{def:QW}
For $\Psi=(\Psi(n))_{n\in\Z}\in\cB$, we define $\cU\Psi\in\cB$ by 
\be\label{eq:TimeEvQW}
(\cU\Psi)(n)=P_{n+1}\Psi(n+1)+Q_{n-1}\Psi(n-1),
\ee
where the matrices $P_n,\,Q_n$ are given by 
$$
P_n:=\begin{pmatrix}1&0\\0&0\end{pmatrix}U_n,\quad
Q_n:=\begin{pmatrix}0&0\\0&1\end{pmatrix}U_n.
$$
\end{definition}

We consider the scattering problem for the quantum walk under the following conditions:

\noindent
\textbf{(U1)}
$U_n\in \cS$ for all $n\in \mathbb Z$. 

\noindent
\textbf{(U2)}
$U_n=I_2$ except for a finite number of $n\in\Z$.

It is not difficult to see that the condition (U1), the non-occurrence of the perfect reflection, implies the existence and the uniqueness of the stationary states, that is, for any $\bm{v}\in\C^2$ and $n_1\in\Z$, there exists a unique state $\Psi=(\Psi(n))_{n\in\Z}\in\cB 
$ satisfying the equation \eqref{eq:StationaryQW} and $\Psi(n_1)=\bm{v}$.

Let $N_0$ be a positive integer such that $U_n=I_2$ for all $|n|\ge N_0$ and $\Psi_\ope{in}^\pm, \Psi_\ope{out}^\pm\in \cB$ be the states satisfying the stationary equation \eqref{eq:StationaryQW} and the conditions
\be
\begin{aligned}
&
\Psi_\ope{in}^-(-N_0)=(0,1)^T,\quad
&\Psi_\ope{out}^-(-N_0)=(1,0)^T,\\
&
\Psi_\ope{in}^+(N_0)=(1,0)^T,\quad
&\Psi_\ope{out}^+(N_0)=(0,1)^T.
\end{aligned}
\ee
Obviously $\Psi_\ope{in}^\pm, \Psi_\ope{out}^\pm$ are independent of the choice of $N_0$. 
We call $\Psi_\ope{in}^\pm$ and $\Psi_\ope{out}^\pm$ {\it incoming} and {\it outgoing} states.

We define the scattering matrix $\SM_\ope{QW}$ for the quantum walk in analogy with the Schr\"odinger case.
\begin{definition}
The scattering matrix $\SM_\ope{QW}$ for the quantum walk is the $2\times 2$ matrix defined by
\be\label{eq:DefSM}
(\Psi_\ope{in}^+,\Psi_\ope{in}^-)=(\Psi_\ope{out}^-,\Psi_\ope{out}^+)\SM_\ope{QW}.
\ee
\end{definition}

\begin{remark}
This scattering matrix $\SM_\ope{QW}$ coincides with the one $\widehat{\Sigma}(\theta)$ defined in \cite{Mo,KKMS} with $\theta=0$. Remark that in \cite{Mo,Su} the condition for $U_n$ with large $|n|$ is  weaker  than $(U2)$. 
But for the purpose of the comparison with the Schr\"odinger equations, it is enough to consider  quantum walks  satisfying $(U2)$. 
\end{remark}
\begin{remark}
The scattering matrix $\SM_\ope{QW}$ is equivalently defined by 
\be
\SM_\ope{QW}=U_{-N_0+1}*U_{-N_0+2}*\cdots*U_{N_0-1}.
\ee
\end{remark}

Next we define the Feynman integral in quantum walk. 
For this, it is convenient to regard $\mathbb Z$ as directed graph.
Let $(n;R):=(n-1,n)$ and $(n;L):=(n+1,n)$ represent the right and left oriented arcs respectively, and set $\cA=\{(n;R),(n;L);\,n\in\Z\}$. We identify $\cB=l^\infty({\mathbb Z};{\mathbb C}^2)$ with $l^\infty(\cA)$ by the map $\iota:\cB\to l^\infty(\cA)$ defined  for $\Psi=((\Psi_1(n),\Psi_2(n))^T)_{n\in\Z}\in\cB$ by
\be\label{eq:IdA}
(\iota\Psi)(n;L)=\Psi_1(n),\quad
(\iota\Psi)(n;R)=\Psi_2(n).
\ee
By this identification, the operator $\cU':=\iota\,\cU\iota^{-1}$ satisfies, for $\psi\in l^\infty(\cA)$,
\be
\begin{aligned}
&(\cU'\psi)(n;L)=a_{n+1}\psi(n+1;L)+b_{n+1}\psi(n+1;R),\\
&(\cU'\psi)(n;R)=c_{n-1}\psi(n-1;L)+a_{n-1}\psi(n-1;R),
\end{aligned}
\ee
where $a_n,\,b_n,\,c_n$ stand for the $(1,1),(1,2),(2,1)$-entries of $U_n$.

Each finite walk $\gamma$ on $\Z$ is a finite consecutive sequence of arcs $\gamma=(A_0,A_1,\ldots,A_{L(\gamma)})$ such that the origin $o(A_l)$ of $A_l$ coincides with the terminus $t(A_{l-1})$ of $A_{l-1}$ for $l$ from 1 to $L(\gamma)$, the length of the walk $\gamma$. 

We associate  with each finite walk $\gamma$ a probability amplitude 
\be\label{eq:DefPA}
\Phi(\gamma):=\prod_{l=1}^{L(\gamma)}\Phi(A_{l-1},A_l),
\ee
where
\be
\Phi(A_{l-1},A_l):=\left\{
\begin{aligned}
&a_{o(A_l)}\quad \text{ if }A_{l-1}=(o(A_l);L),\ A_l=(o(A_l)-1;L),\\
&b_{o(A_l)}\quad \text{ if }A_{l-1}=(o(A_l);R),\ A_l=(o(A_l)-1;L),\\
&c_{o(A_l)}\quad \text{ if }A_{l-1}=(o(A_l);L),\ A_l=(o(A_l)+1;R),\\
&a_{o(A_l)}\quad \text{ if }A_{l-1}=(o(A_l);R),\ A_l=(o(A_l)+1;R).
\end{aligned}\right.
\ee

The following theorem asserts that the entries $s_{jk}$ of the scattering matrix $\SM_\ope{QW}$ is represented by a countable sum of the over all possible walk from $(-1)^{k-1}\infty$ to $(-1)^j\infty$:
\begin{theorem}
Let $G_{jk}$ be the countable set of walks 
$\gamma=(A_0,A_1,\ldots,A_{L(\gamma)})$ with $o(A_0)=(-1)^{k-1}N_0$ and $t(A_{L(\gamma)})=(-1)^jN_0$.
Then the $(j,k)$-entry $s_{jk}$ of $\SM_\ope{QW}$ is given by
the series
\be\label{eq:SumPA}
s_{jk}=
(-1)^{j+k}\sum_{L=1}^\infty
\sum_{\gamma\in G_{jk},L(\gamma)=L}
\Phi(\gamma).
\ee
The sum over $\{\gamma\in G_{jk};L(\gamma)=L\}$ is a finite sum for each positive integer $L$, and  the infinite sum with respect to $L$ is absolutely convergent.
\end{theorem}
\begin{remark}
The total sum $\displaystyle\sum_{\gamma\in G_{jk}}
\Phi(\gamma)$ is not always absolutely convergent.
\end{remark}
\begin{proof} 
Consider the time evolutions of the initial states $\Psi^{(1)}=(\Psi^{(1)}(n))_{n\in\Z}\in\cB$ and $\Psi^{(2)}=(\Psi^{(2)}(n))_{n\in\Z}\in\cB$  given respectively by 
$$
\begin{aligned}
&
\Psi^{(1)}(N_0)=(1,0)^T\quad
\Psi^{(1)}(n)=0\quad(n\in\Z\setminus\{N_0\}),\\
&
\Psi^{(2)}(-N_0)=(0,1)^T\quad
\Psi^{(2)}(n)=0\quad(n\in\Z\setminus\{-N_0\}).
\end{aligned}
$$
It is not difficult to check that for each positive integer $L$, the finite sum $\sum_{\gamma\in G_{jk},L(\gamma)=L}\Phi(\gamma)$ is given by 
\be\label{eq:FiniteSum}
\sum_{\gamma\in G_{jk},L(\gamma)=L}\Phi(\gamma)
=
(\cU^L\Psi^{(k)})_j((-1)^jN_0).
\ee
Lemma \ref{lem:EEsti} below shows the absolute convergence of the sum over $L$. 
It also shows the pointwise convergence of the infinite sum $\sum_{L\in\Z}\cU^L\Psi^{(k)}$ $(k=1,2)$ i.e., the convergence of $\sum_{L\in \Z}(\cU^L\Psi^{(k)})(n)$ for each $n$, and hence the limit is an element of $\cB$. This is clearly a solution to the stationary equation, and we have 
\be
\begin{aligned}
&\sum_{L\in\Z}\cU^L\Psi^{(1)}
=s_{11}\Psi_\ope{out}^-
=\Psi_\ope{in}^+-s_{21}\Psi_\ope{out}^+,\\
&\sum_{L\in\Z}\cU^L\Psi^{(2)}
=\Psi_\ope{in}^--s_{12}\Psi_\ope{out}^-
=s_{22}\Psi_\ope{out}^+.
\end{aligned}
\ee
This with \eqref{eq:FiniteSum} implies the theorem.
\end{proof}

\begin{lemma}\label{lem:EEsti}
There exist $C>0$ and $0<M<1$ such that the inequality
\be
\|(\cU^L\Psi^{(k)})(n)\|_{\C^2}
\le
CL^{2N_0-1}M^L
\ee
holds for any $\left|n\right|\le N_0$ and $L\in\N$.
\end{lemma}
\begin{proof}
Since we consider only finite entries of $\cU^L\Psi^{(k)}$, we introduce the $(4N_0-2)\times(4N_0-2)$-matrix 
$$
E:=\begin{pmatrix}
O				&P_{-N_0+2}	&				&	&\\
Q_{-N_0+1}	&O				&P_{-N_0+3}	&	&\\
				&Q_{-N_0+2}	&O				&\ddots&\\
				&				&\ddots 		&\ddots&P_{N_0-1}\\
				&				&				&Q_{N_0-2}&O
\end{pmatrix}.
$$
Then we have 
$$
\begin{pmatrix}
\cU^L\Psi^{(k)}(-N_0+1)\\\cU^L\Psi^{(k)}(-N_0+2)\\\vdots\\\cU^L\Psi^{(k)}(N_0-1)
\end{pmatrix}
=E^Lv_k,\quad
v_1=\begin{pmatrix}0\\\vdots\\0\\1\\0\end{pmatrix},\ 
v_2=\begin{pmatrix}0\\1\\0\\\vdots\\0\end{pmatrix}\in\C^{4N_0-2}.
$$
A similar argument to \cite[Lemma 4.2]{KKMS} shows that the absolute value of each eigenvalue of $E$ is less than $1$, and that the dimension of each generalized eigenspace is at most $2N_0-1$. 
We denote the maximum of the absolute value of the eigenvalues of $E$ by $M$. Then we have $\|E^Lv\|
\le CL^{2N_0-1}M^L\|v\|
$ for any $v\in\C^{4N_0-2}$, $L\in\N$.
\end{proof}

\begin{figure}
\centering
\includegraphics[bb=0 0 774 157, width=16.5cm]{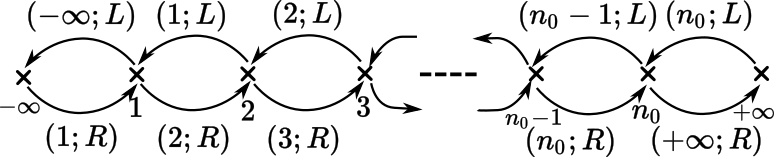}
\caption{Quantum walk on $\cA_{n_0}$}
\label{Fig1}
\end{figure}
Under the assumption (U2), $\Psi(n)$ is independent of $n$ for $n\ge N_0$ and $n\le -N_0$ respectively if $\Psi=(\Psi(n))_{n\in \Z}$ is a solution to the equation \eqref{eq:StationaryQW}.
By regarding $\{n\le-N_0\}$ and $\{n\ge N_0\}$ as infinite tails, such a quantum walk is reduced to the one on the finite subset $\{\left|n\right|\le N_0-1\}$ with two inifinite tails, a particular version of \cite{HiSe} where the authors considered the quantum walk on finite graphs with several infinite tails. Conversely, every quantum walk on a finite subset of $\Z$ can be represented as the one on $\Z$ satisfying the condition (U2). 
For simplicity, we shift $\{\left|n\right|\le N_0-1\}$ to $\{1,2,\ldots,2N_0-1\}$. 

We introduce the quantum walk on $[n_0]:= \{1,2,\ldots,n_0\}$ for an integer $n_0$ greater than 1 by identifying $[n_0]$ with a graph $\mc{G}_{n_0}=(\mc{V}_{n_0},\mc{A}_{n_0})$, where the set of vertices $\mc{V}_{n_0}$ and the set of arcs $\mc{A}_{n_0}$ are given by
\be
\begin{aligned}
&\mc{V}_{n_0}= \{-\infty,1,2,\ldots, n_0, +\infty\},\\
&\mc{A}_{n_0}=\{(n;L), (n;R);\,1\le n\le n_0\}\cup\{(-\infty;L),(+\infty;R)\}.
\end{aligned}
\ee
In this case, the space of the states is the finite dimensional vector space 
$$
\cB:=l^\infty([n_0];\C^2)\oplus l^\infty(\{+\infty,-\infty\};\C)
\cong\C^{2n_0+2},
$$
and the identification $\iota:\cB\to l^\infty(\cA_{n_0})$ is given in a similar way as \eqref{eq:IdA}, that is, 
\be
\begin{aligned}
&(\iota\Psi)(n;L)=\Psi_1(n),\quad
(\iota\Psi)(n;R)=\Psi_2(n),\quad(n=1,2,\ldots,n_0),\\
&(\iota\Psi)(-\infty;L)=\Psi(-\infty),\quad
(\iota\Psi)(+\infty;R)=\Psi(+\infty).
\end{aligned}
\ee
The time evolution $\cU$ is determined by a finite sequence $(U_n)_{n=1}^{n_0}$ of $2\times2$ unitary matrices,  defined by \eqref{eq:TimeEvQW} for $n=2,3,\ldots,n_0-1$, and by
\be\label{eq:DefUonFG}
\begin{aligned}
&(\cU\Psi)(-\infty)
:=(1,0)U_{1}\Psi(1),\quad
(\cU\Psi)(+\infty):=(0,1)U_{n_0}\Psi(n_0),\\
&(\cU\Psi)(1):=P_2\Psi(2)+(0,\Psi_2(1))^T,\\
&(\cU\Psi)(n_0):=(\Psi_1(n_0),0)^T+Q_{n_0-1}\Psi(n_0-1).
\end{aligned}
\ee
Under (U1), the existence and the uniqueness of the stationary states also hold for this case, and we define the scattering matrix $\SM_\ope{QW}$ 
and the probability amplitude $\Phi$ in the same way, \eqref{eq:DefSM}, 
\eqref{eq:DefPA} with $\Z$ replaced by $[n_0]$ and stationary states $\Psi_\ope{in}^\pm$, $\Psi_\ope{out}^\pm$ satisfying
\be
\begin{aligned}
&
(\Psi_\ope{in}^-)_2(1)=1,\ \Psi_\ope{in}^-(-\infty)=0,
&&
(\Psi_\ope{out}^-)_2(1)=0,\ \Psi_\ope{out}^-(-\infty)=1,\\
&
(\Psi_\ope{in}^+)_1(n_0)=1,\ \Psi_\ope{in}^+(+\infty)=0,\ 
&&
(\Psi_\ope{out}^+)_1(n_0)=0,\ \Psi_\ope{out}^+(+\infty)=1.
\end{aligned}
\ee
Then the scattering matrix is computed by the countable sum of probability amplitudes \eqref{eq:SumPA}.

\subsection{Equivalence}
We consider the corresponding quantum walk on $[n_0]$ to the scattering problem given in Sec.~\ref{Sec:Sch} assuming (V1). 
We set $g_0=J_\ope{out}^-$, $g_{n_0}=J_\ope{in}^+$ and $U_n:=\mathcal{M}(T_n)\in\cS$. 
We consider the time evolution $\cU$ determined by the finite sequence $(U_n)_{n=1}^{n_0}$ in the manner of \eqref{eq:DefUonFG}. 
Then we have the following theorem which claims the coincidence of the scattering matrices:
\begin{theorem}\label{thm:SMatrix}
The 
matrices $\SM_\ope{QM}$ and $\SM_\ope{QW}$ coincide with each other.
\end{theorem}

This implies that the scattering matrix $\SM_\ope{QM}$ is obtained by the countable sum of the probability amplitudes of the walks \eqref{eq:SumPA}, or in other words,  
the operation $*$ in $\cS$ can be computed by the time evolution \eqref{eq:FiniteSum} of the quantum walk.

Theorem \ref{thm:SMatrix} is an immediate consequence of the following Lemma \ref{thm:QWQM1}, which relates a solution to the Schr\"odinger equation to a stationary state of the quantum walk. In particular, Jost solutions $J_\ope{in}^\pm$ and $J_\ope{out}^\pm$ correspond to $\Psi_\ope{in}^\pm$ and $\Psi_\ope{out}^\pm$, respectively. 
\begin{lemma}\label{thm:QWQM1}
For a sequence of constant pairs $(\eta_{n,1},\eta_{n,2})_{n=0}^{n_0}\in\C^{2(n_0+1)}$, the following two conditions are equivalent.
\begin{enumerate}
\item  
There exists a global solution $\pphi$ on $\R$ 
to the Schr\"odinger equation \eqref{Scheq} satisfying for each
$n=0,1,\ldots,n_0$
\be\label{eq:EquivSch}
\pphi=\eta_{n,1}g_n+\eta_{n,2}\bar{g}_n.
\ee 
\item 
The state $\Psi\in \cB\cong\C^{2n_0+2}$ given by 
\be\label{eq:EquivQW}
\Psi(n)=
\begin{pmatrix}\eta_{n,1}\\\eta_{n-1,2}\end{pmatrix},\quad 
\Psi(-\infty)=\eta_{0,1},\quad\Psi(+\infty)=\eta_{n_0,2}
\ee 
satisfies the equation $\cU\Psi=\Psi$. 
\end{enumerate}
\end{lemma} 
\begin{remark}Lemma \ref{thm:QWQM1} holds without the condition (V1). 
\end{remark}
\begin{proof}
By definition \eqref{eq:DefTn} of $T_n$, the first condition \eqref{eq:EquivSch} is equivalent to 
$$\begin{pmatrix}\eta_{n-1,1}\\\eta_{n-1,2}\end{pmatrix}
=T_n
\begin{pmatrix}\eta_{n,1}\\\eta_{n,2}\end{pmatrix}\quad
\text{for}\quad n=1,2,\ldots,n_0.
$$
In fact, we have
$$
(g_{n-1},\bar{g}_{n-1})\begin{pmatrix}\eta_{n-1,1}\\\eta_{n-1,2}\end{pmatrix}
=
\pphi
=(g_n,\bar{g}_n)\begin{pmatrix}\eta_{n,1}\\\eta_{n,2}\end{pmatrix}
=(g_{n-1},\bar{g}_{n-1})T_n\begin{pmatrix}\eta_{n,1}\\\eta_{n,2}\end{pmatrix}.
$$
Then by definition 
of $U_n$, this is equivalent to 
$$
\begin{pmatrix}\eta_{n-1,1}\\\eta_{n,2}\end{pmatrix}
=
U_n\begin{pmatrix}\eta_{n,1}\\\eta_{n-1,2}\end{pmatrix}
\quad\text{for}\quad n=1,2,\ldots,n_0.
$$
We can denote this condition by $\cU\Psi=\Psi$ with the state $\Psi$ given by \eqref{eq:EquivQW}.
\end{proof}

\section{Barrier-top scattering and Hadamard walk}\label{Sec:BTandHW}
As a natural and interesting example, we consider the barrier-top scattering for Schr\"odinger equation. We suppose that the potential $V\in C^\infty(\R;\R)$ satisfies the following condition in addition to (V1) and (V2):

\noindent
\textbf{(B1)} The maximum $V_0:=\max_{x\in\R}V(x)$ is positive and the maximum points are all non-degenerate, i.e., $V''(x)<0$ if $V(x)=V_0$.

Following \cite{Ra,FR,Fu},  we also assume the analyticity of the potential (see also \cite{CPS} for the $C^\infty$ case).

\noindent
\textbf{(B2)}
There exists a real number $\theta_0\in(0,\pi/2)$ such that the function $V$ is analytic in the angular complex domain 
$$
\cS_{\theta_0}=\{x\in\C;\,\left|\im x\right|<(1+\left|\re x\right|)\tan\theta_0\}.
$$

We consider energies $\lambda$ close to the barrier-top $V_0$, more precisely, in an interval $\Lambda_h(C):=V_0-[h^M,C_0h]$
for arbitrary $h$-independent constants $C_0>0$,  $M>1$.
Let $V^{-1}(V_0)=\{o_n;\,n=1,2,\ldots,n_0\}$ with $-\infty<o_1<o_2<\cdots<o_{n_0}<+\infty$, $n_0\in\N\setminus\{0\}$ be the maximum points. There exists exactly two real zeros $ x_{2n-1}(\lambda)<o_n<x_{2n}(\lambda)$ of $V(x)-\lambda$ near each $o_n$ for $\lambda\in\Lambda_h(C)$. The interval $K_n$ in section 2.2 is $[x_{2n-1}(\lambda),x_{2n}(\lambda)]$ here.

In each classically allowed interval $I_n=(x_{2n}(\lambda),x_{2n+1}(\lambda))$, $1\le n\le n_0-1$,
there exists a WKB solution $g_n$ to the Schr\"odinger equation 
\eqref{Scheq} whose asymptotic behavior as $h\to +0$ is given by
$$
g_n(x;\lambda,h)=\frac{1+\ord(h)}{\sqrt[4]{\lambda-V(x)}}e^{- i\int_{x_{2n}(\lambda)}^x\sqrt{\lambda-V(y)}\,dy/h}.
$$
Its complex conjugate $\bar g_n$ is also a solution to \eqref{Scheq}
and the pair $(g_n,\bar g_n)$ makes a basis of solutions. Moreover, we can choose $g_n$ so that the wronskian $\W(g_n,\bar{g_n})$ is equal to 
$2i/h$, which is independent of the interval.


In this setting, the semiclassical asymptotics of the transfer matrices $T_n$ were obtained in \cite{Ra,FR,Fu}. In the asymptotics, there appear the action integrals associated with the arcs:
$$
S_{n}(\lambda):=\int_{I_n}\sqrt{\lambda-V(x)}\,dx,
$$
and those associated with the tails $I_-:=(-\infty, x_1)$, $I_+:=(x_{2n_0},+\infty)$:
\begin{align*}
S_-(\lambda)&:=\int_{I_-}\left(\sqrt{\lambda}-\sqrt{\lambda-V(x)}\right)dx-\sqrt{\lambda}x_1, \\
S_+(\lambda)&:=\int_{I_+}\left(\sqrt{\lambda}-\sqrt{\lambda-V(x)}\right)dx+\sqrt{\lambda}x_{2n_0},
\end{align*}
and also the Agmon distance of the barrier
$$
B_n(\lambda):=\int_{K_n}\sqrt{V(x)-\lambda}\,dx.
$$
Remark that $B_n(\lambda)$ is a holomorphic function in a neighborhood of $\lambda=V_0$ where as $S_n(\lambda)$, $S_\pm(\lambda)$ have logarithmic singularity there (see \cite{FR}).
Related with this fact, the asymptotics contain the following function
$$
N_n^\pm(\lambda):=N\left(e^{\pm i\pi/2}\frac{B_n(\lambda)}{\pi h}\right),\quad N(z):=\frac{\sqrt{2\pi}}{\Gamma(1/2+z)}e^{z\log(z/e)}.
$$
The multi-valued function $N(z)$ is called {\it barrier penetration factor} in \cite{LTM}. It tends to $\sqrt{2}$ as $z\to0$ and to $1$ as $\left|z\right|\to+\infty$ in the sector $\left|\arg z\right|<\pi$ by the Stirling's formula. 

Denote
$g_0:=\lambda^{-1/4}J_\ope{in}^-$ and $g_{n_0}:=\lambda^{-1/4}J_\ope{in}^+$. Then the transfer matrices $T_n$ defined by
$(g_{n-1},\bar{g}_{n-1})T_n=(g_n,\bar{g}_n)$ for 
$1\le n\le n_0$ are of the form
$$
T_1=T^-T_{o_1},\ 
T_n=T_{I_{n-1}}T_{o_n}\,(n=2,3,\ldots,n_0-1), \ 
T_{n_0}=T_{I_{n_0-1}}T_{o_{n_0}}(T^+)^{-1},
$$
where the matrices $T^\pm$, $T_{o_n}$, $T_{I_n}$ satisfy the following asymptotics as $h\to +0$ for $\lambda\in\Lambda_h(C)$:
$$
T^\pm(\lambda)=
\begin{pmatrix}
e^{iS_\pm/h}(1+\ord(h))&\ord(e^{-\e/h})\\
\ord(e^{-\e/h})&e^{-iS_\pm/h}(1+\ord(h))
\end{pmatrix},
$$
$$
T_{I_n}(\lambda)=
\begin{pmatrix}
e^{iS_{n}/h}(1+\ord(h))&\ord(e^{-\e/h})\\
\ord(e^{-\e/h})&e^{-iS_{n}/h}(1+\ord(h))
\end{pmatrix},
$$
$$
T_{o_n}(\lambda)=e^{B_n/h}
\begin{pmatrix}
N_n^+(1+\ord(h\log h))&1+\ord(h)\\
1+\ord(h)&N_n^-(1+\ord(h\log h))
\end{pmatrix}.
$$
These asymptotics of transfer matrices are otained by using the reduction to the normal form near the barrier top (see also \cite{HeSj}). 

As we saw in the previous section, we can associate to this quantum mechanics on the line a quantum walk
on the graph $[n_0]$ by giving the unitary operators $U_n$ at each vertex $n=1,2,\ldots,n_0$ by
$U_n:=\mathcal{M}(T_n)\in\mathcal S$.

It is interesting to observe the semiclassical limit of the quantum walk corresponding to this barrier-top scattering.

The principal term $U_n^0=U_n^0(\lambda,h)$ of $U_n$ is given by
\begin{align*}
&U_1^0(\lambda,h)=\frac{e^{(iS_--B_1)/h}}{N_1^-}
\begin{pmatrix}
1&e^{(B_1+iS_-)/h}\\-e^{(B_1-iS_-)/h}&1
\end{pmatrix},\\
&U_n^0(\lambda,h)=\frac{e^{(iS_{n-1}-B_n)/h}}{N_n^-}
\begin{pmatrix}
1&e^{(B_n+iS_{n-1})/h}\\-e^{(B_n-iS_{n-1})/h}&1
\end{pmatrix},\\
&U_{n_0}^0(\lambda,h)=\frac{e^{(i(S_{n_0-1}-S_+)-B_{n_0})/h}}{N_{n_0}^-}
\begin{pmatrix}
1&e^{(B_{n_0}+i(S_{n_0-1}+S_+))/h}\\-e^{(B_{n_0}-i(S_{n_0-1}+S_+))/h}&1
\end{pmatrix}.
\end{align*}
In particular, for $\lambda\in\Lambda_h(C)$ satisfying $V_0-\lambda=o(h)$, we have
\begin{align*}
&U_1^0(\lambda,h)=\frac{e^{iS_-/h}}{\sqrt{2}}
\begin{pmatrix}
1&e^{iS_-/h}\\-e^{-iS_-/h}&1
\end{pmatrix},\\ 
&U_n^0(\lambda,h)=\frac{e^{iS_{n-1}/h}}{\sqrt{2}}
\begin{pmatrix}
1&e^{iS_{n-1}/h}\\-e^{-iS_{n-1}/h}&1
\end{pmatrix},\\ 
&U_{n_0}^0(\lambda,h)=\frac{e^{i(S_{n_0-1}-S_+)/h}}{\sqrt{2}}
\begin{pmatrix}
1&e^{i(S_{n_0-1}+S_+)/h}\\-e^{-i(S_{n_0-1}+S_+)/h}&1
\end{pmatrix}.
\end{align*}
This is a version of the Hadamard walk, that is, the square of the absolute value of each entry is $1/2$. 
The value of each entries is depend analytically on $\lambda$ in some complex neighborhood of $\Lambda_h(C)$, and the rate of local transmission coefficient to reflection coefficient, the one of the square of the absolute values of $(1,1),(2,2)$-entry to $(1,2),(2,1)$-entry of $U_n$, monotonically decreasing as $\lambda\to \min\Lambda_h(C)=V_0-Ch$.

\section*{Acknowledgements}
The author is grateful to Ritsumeikan University for the finantial support, KENKYU-SHOREI Scholarship A.

Kenta Higuchi, 
Department of Mathematical Sciences, 
Ritsumeikan University, 
1-1-1 Noji-Higashi, Kusatsu, 
525-8577,  Japan

\textit{E-mail address}: ra0039vv@ed.ritsumei.ac.jp

\end{document}